\documentclass[reprint,amsmath,amssymb,aps,pra,twocolumn]{revtex4-2}

\usepackage{dcolumn}
\usepackage{bm}

\usepackage{relsize}
\usepackage[utf8]{inputenc}
\usepackage[english]{babel}
\usepackage[T1]{fontenc}
\usepackage{hyperref}
\usepackage{graphicx}

\usepackage{amsthm}
\usepackage{amsmath} 
\usepackage{amssymb}
\usepackage[colorinlistoftodos]{todonotes}
\usepackage{ulem}
 
\newtheorem{theorem}{Theorem}

\newtheorem{lemma}[theorem]{Lemma}
\newtheorem{definition}[theorem]{Definition}

\newcommand{\edit}[1]{{\color{red}#1 }}

\begin{document}

\preprint{APS/123-QED}

\title{New local stabilizer codes from local classical codes}

\author{Lane G. Gunderman}
 \email{lanegunderman@gmail.com}
\affiliation{Department of Electrical and Computer Engineering,
                    University of Illinois Chicago,
                   Chicago, Illinois, 60607}

\date{\today}

\begin{abstract}
Amongst quantum error-correcting codes the surface code has remained of particular promise as it has local and very low-weight checks, even despite only encoding a single logical qubit no matter the lattice size. In this work we discuss new local and low-weight stabilizer codes which are obtained from the recent progress in $2D$ local classical codes. Of note, we construct codes with weight and qubit use count of $5$ while being able to protect the information with high distance, or greater logical count. We also consider the Fibonacci code family which generates weight and qubit use count of $6$ while having parameters $[[O(l^3),O(l),\Omega(l)]]$. While other weight-reduction methods centered on lowering the weight without regard to locality, this work achieves very low-weight and geometric locality. This work is exhaustive over translated classical generators of size $3\times 3$ and up to size $17\times 17$ classical bit grids.
\end{abstract}

\maketitle

\section{Introduction}

Quantum mechanical interactions are generated through geometrically local interactions, be it between photons, electron orbitals, or other modes. This naturally motivated means of protecting quantum information using checks that are also geometrically local. Besides this, another natural motivation is to use fewer body interactions as precise control of many-body interactions becomes increasingly challenging as the number of interactions increases. This latter aim led to ideal methods for protecting quantum information requiring few body interactions, which in terms of quantum error-correcting codes corresponds to low-weight checks. These pair of restrictions led to the dominance of the surface code in quantum error-correcting codes as the checks are localized in $2D$ space using $4$ neighboring additional qubit interactions in each check \cite{kitaev2006anyons}. While these properties of the code permit it to have a high threshold, and recently experimentally achieve prolonged storage of information, no matter the size of the surface patch used only a single logical qubit is encoded always \cite{acharya2024quantum}. This low logical encoding then requires lattice surgery to perform multi-qubit operations, which can increase the overhead a fair amount \cite{horsman2012surface,litinski2019game}. There exist non-Euclidean lattices which provide for more encoding, however, their physical implementability for some platforms would be very challenging \cite{breuckmann2016constructions}.

A radically different approach is to drastically increase the encoding rate while retaining the low-weight property of the checks, but permitting the geometric locality to be lost. These result in non-local quantum low-density parity-check (qLDPC) codes \footnote{Formally speaking, the surface code is a qLDPC code, but here we mean a non-localized code}. While these can have low-weight and higher encoding rate, the fact that the checks are so spread out can make performing these checks challenging. The very best qLDPC codes even require a large amount of far range checks \cite{bravyi2010tradeoffs,baspin2024improved,dai2024locality}.

Given this difficult trade-off, a myriad of different approaches have evolved which aim to keep locality while reducing the weight of the checks--amongst these approaches include hyperboic lattice codes \cite{breuckmann2016constructions}, subsystem codes \cite{aliferis2007subsystem,bravyi2012subsystem,higgott2021subsystem}, floquet codes \cite{hastings2021dynamically,higgott2024constructions}, and weight reduction methods \cite{hastings2021quantum,sabo2024weight,he2025discovering}. Here we restrict ourselves to true stabilizer codes (the quantum analog of classical additive codes), which are believed to be the most promising approach and have a number of nice fault-tolerant features known for them. In this work, we primarily focus on smaller code instances which: 1) are purely geometrically local, 2) have very low weight, 3) encode more than a single logical register. This work leverages the recent results on $2D$ local classical codes to obtain these results, however, we provide some further analysis of the codes which are possible and illustrate their possible great utility for protecting quantum information \cite{ruiz2025ldpc}.

The work is organized as follows. In the next section definitions are laid out. Following this we show our construction and provide various instances with particularly appealing properties as well as families for which this method applies, making comparisons with other promising results. We then close out by indicating possible future directions to extend this work.

\section{Definitions}

This work focuses on balancing locality, weight, and the parameters of quantum codes. Since this is achieved through using classical codes, we begin by defining the parameters for classical additive error-correcting codes. The parameters are a triplet of values: $n$ specifies how many classical registers (bits) are used, $k$ the number of logical registers (bits) which are bases for the null space of the parity check matrix, and $d$ the distance which is the minimal number of columns which are linearly dependent. These are written as $[n,k,d]$.

In the quantum case we assume that the given codes have been written in the symplectic representation so that the codes are represented as a matrix of powers of the Pauli operators \cite{lidar2013quantum}. Using this representation, the parameters for the quantum code are rather similar to the classical code: $n$ is the number of physical registers (qubits, or more general), $k$ is the number of logical registers which can be deduced from the number of physical registers minus the count of the symplectically independent rows of the check matrix, and $d$ is the distance of the code which is the smallest number of columns that are symplectically linearly dependent, but does not correspond to one of the checks. These are written as $[[n,k,d]]$.

A final element needed for the results in this work is our selected method to transform classical codes into quantum codes. Let $H_1$ be a classical linear code with $r_1$ rows and parameters $[n_1,k_1,d_1]$ and let $H_2$ be a classical linear code with $r_2$ rows and parameters $[n_2,k_2,d_2]$. Then we may form the hypergraph product stabilizer code from these codes with parity checks given by $H_X=(H_1\otimes I_{n_2}\quad I_{r_1}\otimes H_2^T)$ and $H_Z=(I_{n_1}\otimes H_2\quad H_1^T\otimes I_{r_2})$ \footnote{For the qudit version, or generally local-dimension-invariant form, we simply put a minus sign in front of one term so that the symplectic product is perfectly zero.}\cite{tillich2013quantum}. The resulting code will have parameters $[[n_1n_2+r_1r_2,k_1k_2+k_1^Tk_2^T,\min(d_1,d_2,d_1^T,d_2^T)]]$.

\subsection{Notation}

Throughout this work we use alphabetically ordered strings of the letters $a$ to $i$ to indicate the cyclically generating local $2D$ classical checks where the letters indicate support within the following $3\times 3$ fundamental check block:
\begin{equation}
    \begin{pmatrix}
        a &b&c\\
        d&e&f\\
        g&h&i
    \end{pmatrix}.
\end{equation}
As a simple example the string $"ab"$ would indicate a set of repetition codes whereby adjacent registers are used in the check, then $"ab"$ is translated across the $2D$ set of registers to generate the full set of checks. When discussing particular instances we augment this notation for generators with a subscript indicating the dimension of the grid of registers in the form $w\times h$ (width by height), and following the string with the traditional classical error-correcting code parameter notation $[n,k,d]$. In this work we will only focus on pairing local $2D$ codes with the repetition code in a hypergraph product construction, although other $1D$ local codes could be employed, such as cyclic codes, however, this would increase the weight and the spatial spread a bit, and so we leave that as a future direction instead. The quantum code derived from using the given classical code paired with the repetition code will be denoted by a right arrow $\rightarrow$ followed by the traditional denotation of the code's parameters $[[n,k,d]]$. As an example, $"cdg"_{8\times 10}[80,10,27]\rightarrow [[2420,10,27]]$. 

Additional notations include a breakdown of the distance into the pair of values $(d_x,d_z)$ indicating the protection against bit and phase flip errors, which is particularly useful for biased noise systems where error rates are unequal. In the work introducing the $2D$ local \textit{classical} codes, they effectively select one distance to be $1$, as they are working with cat qubits which are highly noise biased qubits \cite{ruiz2025ldpc}. Other platforms also exhibit biased noise, including biased erasures, however, we leave further analysis to future work.

Beyond this, we define the check weight of a hypergraph product code as the stabilizer generator with the most nonzero row entries in the symplectic representation of the code and denote it for a stabilizer code by $w$. For the codes considered in this work these will be the classical check's weight plus $2$, as we only consider using the repetition code. We also define a pair of qubit use weight, or column weight. As we are working with CSS codes we can define the $X$-check column weight and the $Z$-check column weight as being respectively the maximal count of nonzero entries within the $X$, and $Z$, supported stabilizer generators. These are denoted by $q_x$ and $q_z$ and in this work are given by the check weight of the $2D$ local classical code. The last weight we define is the total register use count, or total column weight, which is the maximal Hamming weight of the symplectic columns (meaning that in the symplectic representation we add the weight of column $i$ and $i+n$ since they correspond to the same physical register), which we denote by $q$. In this setting $q$ is given by the weight of the classical $2D$ code plus $2$ as the physical registers within the $X$ stabilizer are the local $2D$ checks are supported in the $Z$ stabilizers by the repetition code, and the transpose is true for the other registers. Aggregating this, for the code specified by $"cdg"_{8\times 10}[80,10,27]\rightarrow [[2420,10,27]]$, we have $w=5$, $q_x=q_z=3$ and $q=5$. Of particular note this beats the weights from \cite{hastings2021quantum,sabo2024weight} while also being local. Unfortunately, asymptotically the performance of these codes will be limited, thus not purely outperforming these prior results.

\section{Construction details}

Here we outline some of the details and perks of the construction employed in this work. While the short version is that they simply employ a classical local $2D$ code along with a repetition code as inputs for a hypergraph product code, we provide more details in what follows.


As a first observation, we may consider the surface (or toric) code, which is geometrically local and has low check weights. This code results from the hypergraph product of two repetition codes, which are classical local $1D$ codes. This construction retains the locality and sparsity as the classical codes being used also are such. This motivates the following observation: if a classical code has sparse and local check operators, the transpose code will also have these properties. This is evident by considering the Tanner graph and swapping the roles of the checks and bits. Importantly this holds beyond the $1D$ case.

As our next observation, the hypergraph product construction uses the Cartesian product of the Tanner graphs for the classical codes to generate the qubit patches and the checks, a nice example is seen in \cite{manes2025distance}. This Cartesian product also preserves sparsity and locality. Then if we have a classical local $2D$ code, such as those studied in \cite{ruiz2025ldpc}, we may pair that with a classical local $1D$ code, herein the repetition code $[n,1,n]$ using $n-1$ generators, to generate a pair of $3D$ qubit patches. The checks will be bi-local in that they will use a pair of locally supported checks--one part of the support in the bit/bit patch and one part of the support in the check/check patch, as seen in Figure \ref{fig:enter-label}. We may then complete our construction by inter-laying these $3D$ patches so that the checks are now mono-local. Another notable feature of this construction is that the classical $2D$ code decoder suffices for the decoding problem, meaning that these codes are also likely easier on the decoding front \cite{golowich2024decoding}.

While that completes the construction there are still other aspects to take into consideration. Firstly, we ought to consider the differences in the sizes of the bit/bit qubit patch and the check/check qubit patch. When these differ significantly the spatial range of the checks may be far longer than would be physically appropriate for platforms such as superconducting devices. This can be greatly reduced when the classical codes used not only satisfy the locality constraint but also are low-rate so that the patches are roughly the same size and true constant length lattice locality is achieved. On the other hand this consideration may be of far lower importance for trapped ion and neutral atom systems. In such systems whole lattices corresponding to the different patches might be better suited to be translated as a group. In photonic based devices, this consideration can likely be effectively ignored by different timing of checks or fiber optics routing.

With our construction provided and details sufficiently discussed, we now turn to outlining a variety of examples.

\begin{figure}[t]
    \centering
    \includegraphics[width=\linewidth]{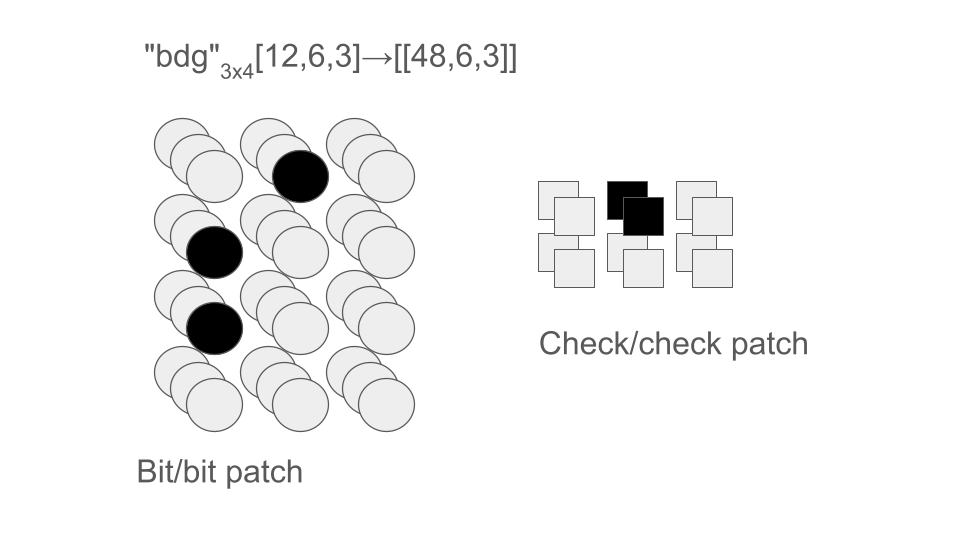}
    \caption{An example code constructed from the generator $"bdg"$ making a $[[48,6,3]]$ code with check weight and qubit use count of $5$ each. As discussed in the text, these patches would be interwoven into a single $3D$ code, and as these differ in dimensions, either translation would be needed or longer spatial range interactions would be required.}
    \label{fig:enter-label}
\end{figure}

\section{Local Code Examples}

\subsection{Repetition-like families}

As our first example family we consider the codes generated by generator patches $"ab"_{w\times h}[wh,h,w]\rightarrow [[whd+(wh-h)(d-1),h,(w,d)]]$. The weight is $4$ for both check weight and qubit usage. This is the full description of the parameters, however, for interpretation, we begin by considering the classical code generated by $"ab"$. This code is equivalent to $h$ copies of a $[w,1,w]$ repetition code. Upon pairing this with a separate repetition code $[d,1,d]$ in a hypergraph product construction, this becomes a stack of (perhaps non-symmetric) surface codes. 


\subsection{The Fibonacci Code}

As an asymptotic family, let us consider the Fibonacci code as introduced in \cite{yoshida2013exotic}. This classical local $2D$ code had parameters going as $[O(l^2),O(l),\Omega(l)]$ with classical check operators of weight $4$. A decoder for this classical code was studied in \cite{nixon2021correcting}, which paired with the recent result in \cite{golowich2024decoding} means that our hypergraph product construction of a local $3D$ quantum code will benefit from the improved decoder therein. The stabilizer code obtained from using a Fibonacci code with a repetition code in the hypergraph product construction generates a code with parameters $[[O(l^3),O(l),\Omega(l)]]$, assuming a repetition code is chosen to equate $d_z$ with $d_x$. This code has check weight $6$ and qubit usage number also of $6$. This code at least saturates the $2D$ local code optimality bound of $kd^2/n=\Omega(1)$ (it could exceed it as the distance is not tightly known for the Fibonacci code), but does not saturate the $3D$ bound as the Layer code achieves \cite{williamson2024layer}. While this does not saturate the $3D$ code bound, this code has spatially small geometric locality, and thus is likely a more practical option than the Layer code families.


\subsection{Selected instances}

In the tables we highlight many code examples. In Table \ref{biasedtable} we list the parameters which would be obtained using the specially tailored $2D$ codes from \cite{ruiz2025ldpc}, taking care to permit the length of the repetition code be specified later so that these are particularly suited for biased-noise systems (using the repetition code for the less noise-prone error-axis). As these have been specially tailored, they are promising early candidates. In Table \ref{kdbyn} we list out of the optimal values obtained for $kd^2/n$ upon considering grids of size up to $17\times 17$ and using the translationally generated generators--where ties existed an automata generator was selected. In this table we see a massive increase in the ratio $kd^2/n$, a common metric for comparison with surface codes, for $2D$ generators of weight $3$ and $4$, but then the marginal gain decreases. Lastly, in Table \ref{ratios} we provide examples of codes which obtain the best $d/n$ and $k/n$ ratios for the resultant hypergraph product code, separated by the check and qubit use weights and selecting for $d\geq 3$ and $k\geq 2$. It is also worth noting that a number of the codes constructed outperform the parameters from \cite{sabo2024weight} while having the same or lower weights and geometric locality--albeit the asymptotic behavior for the codes considered here is worse.

\section{Conclusion and Future Directions}

In this work we have discussed methods and details related to transforming classical $2D$ local codes into quantum error correcting codes with local checks. As the resulting patches may not be of equal size, at times these local, low-weight checks might not be spatially local but merely geometrically local. We also discussed how these can leverage the classical decoders to obtain more rapid decoding. Likewise, this work has primarily analyzed the parameter advantage which can be obtained from the codes constructed, however, it would be of great value to see how great the advantage is for different platforms. As the precise capabilities of platforms vary, providing a numeric expression for the advantage of the locality of these codes is challenging. Additionally methods for performing logical operations and other fault-tolerant properties must be investigated to determine the comparative advantage of these codes. For instance, in some devices spatial swaps of physical qubits are performed effectively perfectly, while in others in can be of tremendous cost. As such, we leave this to further investigation. We believe that these codes could prove to be very beneficial and opens the set of options for local quantum error-correcting codes.











\if{false}
The following Lemma forms our central observation which will enable our main result:

\begin{lemma}
Let $H$ be a classical error-correcting code that is geometrically local, then the transpose code $H^T$ is also geometrically local.
\end{lemma}

\begin{proof}
A geometrically local code's checks will involve some small constant which bounds the weight within a volume given by the extent radius $r$. As such the number of checks using any given register will be bounded by a small constant that is a function of $r$ and the spatial dimensionality of the code, and can be sorted as nearby parity checks, and so $H^T$ will also be geometrically local.
\end{proof}

This could also be seen by considering the transpose of the Tanner graph.
\fi

\begin{widetext}

\begin{table}[h]
    \centering
    \begin{tabular}{|c|c|c|c|c|c|c|} \hline 
        2D seed family & $n$ & $k$ & $d_x$ & $d_z$ & check weight & qubit weight\\ \hline 
         $[20+4l,10+2l,5]$ from \cite{ruiz2025ldpc} & $d(20+4l)+(d-1)(10+2l)$ & $10+2l$ & $d$ & $5$ & $6$ & $6$\\ \hline 
         $[55+5l,22+2l,9]$ from \cite{ruiz2025ldpc} & $d(55+5l)+(d-1)(33+3l)$ & $22+2l$ & $d$ & $9$ & $6$ & $6$\\ \hline 
         $[78+6l,26+2l,12]$ from \cite{ruiz2025ldpc} & $d(78+6l)+(d-1)(52+4l)$ & $26+2l$ & $d$ & $12$ & $6$ & $6$\\ \hline 
         $[119+7l,34+2l,16]$ from \cite{ruiz2025ldpc} & $d(119+7l)+(d-1)(85+5l)$ & $32+2l$ & $d$ & $16$ & $6$ & $6$\\ \hline 
         $[136+8l,34+2l,22]$ from \cite{ruiz2025ldpc} & $d(136+8l)+(d-1)(102+6l)$ & $34+2l$ & $d$ & $22$ & $6$ & $6$\\ \hline
         $[O(l^2),l,\Omega(l)]$ Fibonacci code from \cite{yoshida2013exotic} & $O(l^2d)$ & $O(l)$ & $d$ & $\Omega(l)$ & $6$ & $6$\\ \hline
    \end{tabular}
    \caption{Parameters for bi-local codes where we pair local $2D$ codes with a repetition code. One could pair with other local $1D$ codes, but we focus on this here.}
    \label{biasedtable}
\end{table}

\begin{table}[h]
    \centering
    \begin{tabular}{|c|c|c|c|c|c|c|} \hline 
        2D generator and resulting HGP parameters & (width,height) for $2D$ code & weight $w$ & qubit use $q$ & $\frac{kd^2}{n}$\\ \hline 
        $"be"_{17\times 5}[85,5,17]\rightarrow[[1525,5,17]]$ & $(17,5)$ & $4$ & $4$ & $0.948$\\ \hline
         $"cdg"_{5\times 16}[80,10,27]\rightarrow [[2420,10,27]]$ & $(5,16)$ & $5$ & $5$ & $3.012$\\ \hline 
         $"bdfg"_{17\times 17}[289,34,68]\rightarrow[[21930,34,68]]$ & $(17,17)$ & $6$ & $6$ & $7.169$\\ \hline
         $"cdghi"_{17\times 17}[289,34,75]\rightarrow[[24191,34,75]]$ & $(17,17)$ & $7$ & $7$ & $7.906$\\ \hline
         $"adfghi"_{17\times 16}[272,34,76]\rightarrow[[23222,34,76]]$ & $(17,16)$ & $8$ & $8$ & $8.457$\\ \hline
         $"abcdghi"_{17\times 17}[289,34,83]\rightarrow[[26775,34,83]]$ & $(17,17)$ & $9$ & $9$ & $8.748$\\ \hline
         $"abcdfghi"_{17\times 16}[272,34,80]\rightarrow [[24446,34,80]]$ & $(17,16)$ & $10$ & $10$ & $8.901$\\ \hline
    \end{tabular}
    \caption{In this table we provide a comparison of the codes in this work against the surface code. The surface code has $\frac{kd^2}{n}=\frac{1}{2}$. This table uses the largest value example code, selecting automata codes and smallest examples where equivalent, and is exhaustive over all $3\times 3$ generating patches and up to size $17\times 17$.}
    \label{kdbyn}
\end{table}

\begin{table}[h]
    \centering
    \begin{tabular}{|c|c|c|c|c|c|c|} \hline 
        2D generator and resulting HGP parameters & $k/n$ & $d/n$ & weight $w$ & qubit use $q$\\ \hline 
        Surface code $[[2d^2,1,d]]$ & $\frac{1}{2d^2}$ & $\frac{1}{2d}$ & 4 & 4\\ \hline
        $"ad"_{3\times 3}[9,3,3]\rightarrow [[33,3,3]]$ & $0.091$ & $0.091$ & $4$ & $4$\\ \hline
        $"bdg"_{3\times 4}[12,6,3]\rightarrow [[48,6,3]]$ & $0.125$ & $0.062$ & $5$ & $5$\\ \hline
        $"cde"_{3\times 4}[12,3,3]\rightarrow [[42,3,3]]$ & $0.071$ & $0.071$ & $5$ & $5$\\ \hline
        $"achi"_{5\times 4}[20,12,3]\rightarrow [[84,12,3]]$ & $0.143$ & $0.036$ & $6$ & $6$\\ \hline
        $"abde"_{3\times 3}[9,5,3]\rightarrow [[37,5,3]]$ & $0.135$ & $0.081$ & $6$ & $6$\\ \hline
        $"bdfgh"_{5\times 3}[15,10,3]\rightarrow [[65,10,3]]$ & $0.154$ & $0.046$ & $7$ & $7$\\ \hline
    \end{tabular}
    \caption{In this table we provide a comparison of the codes in this work against the surface code. Here we provide some instances where attention has been given to achieving large ratios for either $k/n$ or $d/n$ for the resulting hypergraph product code. To keep the resulting codes of note, we restrict to $k\geq 2$ and $d\geq 3$. This table selects automata codes and smallest examples where equivalent, and is exhaustive over all $3\times 3$ generating patches for classical codes of size up to $17\times 17$.}
    \label{ratios}
\end{table}
\end{widetext}




\if{false}
The initial idea for a quantum computer was argued as a tool for simulating quantum systems \cite{feynman2018simulating}. This problem was advanced by Lloyd proposing a method for actually simulating quantum systems by approximating it with products of operators \cite{lloyd1996universal}. Unfortunately such products have very poor error-scaling unless higher order product series are used such as those of Suzuki \cite{suzuki1990fractal,suzuki1991general}. Additionally, this seminal work by Lloyd did not specify what sort of operations would be used in the product formula. The simplest building block operations for quantum devices are the Pauli operators, so if the system can be simulated using these operations the feasibility of utilizing some sort of near-term advantage would be improved, providing a useful algorithm for NISQ (noisy intermediate-scale quantum) devices \cite{preskill2018quantum}. A variety of transformations have been devised for transforming Hamiltonians for quantum systems, such as molecules or field-theories, into Pauli operators with associated weights \cite{bravyi2002fermionic,tranter2018comparison,mcclean2014exploiting,mcclean2020openfermion}. Well controlled qubits, especially logical ones, with long coherence time are a precious resource, especially in more near-term quantum devices. These all perform reasonably well in this regard, however, this work solves one open question related to these transformations: for a particular collection of Pauli operators obtained from these transformations, or from some other setting, what is the minimal number of qubits needed in order to perform such operations? Our answer is found from quantum error-correction techniques. Additionally, our solution finds the obtainable lower-bound and also provides an efficient methodology for transforming the original collection of Pauli operators into an equivalent, under Clifford operations, collection using the minimal number of qubits nontrivially. \edit{The equivalence considered in this work is under Clifford conjugations.} To help avoid confusion for this work achieving Pauli operators with minimal weights, or number of non-trivial operations within each multi-qubit Pauli operator, instead we this say that this result minimizes the number of registers, or physical qubits that would be needed to perform the operations. Note that this work has more general applications--wherever Pauli operators appear (from quantum complexity theory to non-local games to fault-tolerance to communication protocols)--but we focus here on the setting of more near-term algorithms related to simulating quantum systems. One way to interpret this paper is as determining an optimal set of Pauli operators resulting from a Clifford circuit \cite{gottesman1998heisenberg,aaronson2004improved}. The optimality achieved is in the number of total qubits, registers, needed in the quantum device to obtain the same circuit as the original given set. The application of this result to the simulation of quantum systems yields the optimal compression from Fermionic modes to qubits, such as that found in \cite{jiang2020optimal,harrison2022reducing}, however, this result extends more broadly as it is just an equivalence between Pauli operators.


Formally, quantum computers operate by performing unitary operators. The rudimentary operations for qubits are the Pauli operators, which act as follows on the computational basis states $\{|0\rangle,|1\rangle\}$:
\begin{equation}
    X|j\rangle=|j\oplus 1\rangle,\quad Z|j\rangle =(-1)^j|j\rangle,\quad j\in \{0,1\}.
\end{equation}
There is one final non-trivial operator, $Y$, which we take as the product of $X$ and $Z$ \footnote{While $XZ= -ZX$ and $Y=iXZ$, for the sake of our arguments it suffices to note that $Y$ is just some scalar times the product}. \if{false}As matrices these are explicitly:
\begin{equation}
    X=\begin{bmatrix}
    0 & 1\\
    1 & 0
    \end{bmatrix},\ Y=\begin{bmatrix}
    0 & -i\\
    i & 0
    \end{bmatrix},\ Z=\begin{bmatrix}
    1 & 0\\
    0 & -1
    \end{bmatrix}.
\end{equation}\fi
These Pauli operators form a basis for Hermitian $2\times 2$ matrices over $\mathbb{C}$. The $n$-fold tensor product of Pauli operators is indicated by $\mathbb{P}^n$, and likewise these $n$-fold tensor products of Pauli operators form a basis for $2^n\times 2^n$ Hermitian matrices over $\mathbb{C}$, which in turn generates the full Lie group for $n$ qubits, meaning that all Unitary transformations are generated from the exponentiation of sums of these operators. Assuming the rotation does not correspond to a rotation about a separable Pauli axis, physically implementing a rotation which corresponds to a sum of these operators is generally not possible. If all the Pauli operators commute, and the Clifford group is easy to physically implement, these Pauli operators can be diagonalized, as well as generally finding Clifford equivalence circuits, efficiently \cite{gottesman1998heisenberg, aaronson2004improved}. In this work we are not restricting ourselves to commuting Pauli operators only so to approximate the exponentiation of sums of non-commuting Pauli operators this could be accomplished with Trotter-Suzuki series. Let $\{A_i\}$ be a collection of Pauli operators with associated weights $\{w_i\}$, then we may approximate $e^{x\sum_{i=1}^n w_iA_i}$ with:
\begin{equation}
            S_1^*(x)=(\prod_{i=1}^n e^{\frac{1}{2}xw_iA_i})(\prod_{i=n}^1 e^{\frac{1}{2}xw_iA_i}),
        \end{equation}
where the first product has the operators in forward order, while the second one has them in reverse order, and generating the further expressions with:
\begin{widetext}
        \begin{eqnarray}
            S_{2m}^*(x)&=&S_{2m-1}^*(x)\\
            &=&[S_{2m-3}^*(p_m x)]^2S_{2m-3}^*((1-4p_m)x)[S_{2m-3}^*(p_m x)]^2,\quad p_m=(4-4^{1/(2m-1)})^{-1},
        \end{eqnarray}
        \end{widetext}
where the equations come from \cite{suzuki1990fractal,suzuki1991general}. This is correct up to $O(x^{2m+1})$, with the remaining error being related to commutators of the operators, so in principle can generate arbitrary operations arbitrarily accurately. This is not the focus of this work, but provides a brief setting to consider the results of this work.

Returning to the Pauli operators on $n$ qubits, working with these non-commuting operators directly is rather cumbersome and so we utilize the symplectic representation for the Pauli operators to transform the tool kit from those of non-commuting group theory to those of linear algebra over $\mathbb{Z}_2$ \cite{nielsen2002quantum,ketkar2006nonbinary}. This linearity of the mapped versions will help with some of the results, reducing them to linear algebraic spaces. 

\begin{definition}
The binary symplectic representation of a Pauli operator is the mapping, $\phi$, which maps $\mathbb{P}^n\mapsto \mathbb{Z}_2^{2n}$. It takes the $i$-th register in the Pauli and sets the $i$-th position in the binary vector to the power of the $X$ operator at that position and the $i+n$-th position to the power of the $Z$ operator. In short, this performs:
\begin{equation}
    \phi \left(\bigotimes_{t=1}^n X^{a_t}Z^{b_t}\right)=\left(\bigoplus_{t=1}^n a_t\right)\bigoplus \left(\bigoplus_{t=1}^n b_t\right).
\end{equation}
This mapping observes:
\begin{equation}
    \phi(s_i\circ s_j)=\phi(s_i)\oplus \phi(s_j)
\end{equation}
where $s_i$ and $s_j$ are two Pauli operators in $\mathbb{P}^n$ and $\oplus$ is componentwise vector addition mod two.
\end{definition}
In this work we will be particularly interested in the commutation relations between Pauli operators, which in the $\phi$ representation is computed as follows:
\begin{definition}
The commutator, formally the symplectic product, between two Pauli operators $\phi(p_1)=(\vec{a}_1|\vec{b}_1)$ and $\phi(p_2)=(\vec{a}_2|\vec{b}_2)$ is computed as:
\begin{equation}
    \phi(p_1)\odot \phi(p_2)=\vec{a}_1\cdot \vec{b}_2-\vec{b}_2\cdot \vec{a}_2 \mod 2
\end{equation}
\end{definition}
For qubit Pauli operators this is $0$ if $p_1$ and $p_2$ commute, and $1$ if not (they anti-commute). Generally this provides the power of the corresponding root of unity induced from swapping the order of the Pauli operators if considering the qudit case as well \cite{ketkar2006nonbinary}.

For convenience we also introduce the following definition to specify the number of compositionally independent Pauli operators contained in our given collection of Pauli operators. \edit{A Pauli operator is compositionally independent of a collection of Pauli operators if there is no way to multiply members of the collection to obtain that operator. A set of compositionally independent Pauli operators are generators for a subgroup of Pauli operators.}
\begin{definition}
For a collection of Pauli operators $\mathcal{P}$ we denote by $rank_\phi(\mathcal{P})$ the number of compositionally independent generators for $\mathcal{P}$, we call this the $\phi$-rank of $\mathcal{P}$.
\end{definition}

\if{false}The $\phi$-rank of $\mathcal{P}$ is simply the rank of the symplectic representation of $\mathcal{P}$.\fi Since there will be a second rank of importance in this work, we have made this separate distinction. In this setting we are only allowed Clifford operations on our collection of Pauli operators. With these preliminary definitions, we are prepared to formally state the problem solved in this work as well as lay our groundwork for proving this result.

\fi

\section*{Acknowledgments}

This work would not have been possible without the availability of the results from \cite{ruiz2025ldpc}. Further, we thank Diego Ruiz for helpful early comments and Mike Vasmer for a very insightful discussion.



\bibliographystyle{unsrt}
\phantomsection  
\renewcommand*{\bibname}{References}

\bibliography{main}

\begin{thebibliography}{10}

\bibitem{kitaev2006anyons}
Alexei Kitaev.
\newblock Anyons in an exactly solved model and beyond.
\newblock {\em Annals of Physics}, 321(1):2--111, 2006.

\bibitem{acharya2024quantum}
Rajeev Acharya, Laleh Aghababaie-Beni, Igor Aleiner, Trond~I Andersen, Markus Ansmann, Frank Arute, Kunal Arya, Abraham Asfaw, Nikita Astrakhantsev, Juan Atalaya, et~al.
\newblock Quantum error correction below the surface code threshold.
\newblock {\em arXiv preprint arXiv:2408.13687}, 2024.

\bibitem{horsman2012surface}
Dominic Horsman, Austin~G Fowler, Simon Devitt, and Rodney Van~Meter.
\newblock Surface code quantum computing by lattice surgery.
\newblock {\em New Journal of Physics}, 14(12):123011, 2012.

\bibitem{litinski2019game}
Daniel Litinski.
\newblock A game of surface codes: Large-scale quantum computing with lattice surgery.
\newblock {\em Quantum}, 3:128, 2019.

\bibitem{breuckmann2016constructions}
Nikolas~P Breuckmann and Barbara~M Terhal.
\newblock Constructions and noise threshold of hyperbolic surface codes.
\newblock {\em IEEE transactions on Information Theory}, 62(6):3731--3744, 2016.

\bibitem{Note1}
Formally speaking, the surface code is a qLDPC code, but here we mean a non-localized code.

\bibitem{bravyi2010tradeoffs}
Sergey Bravyi, David Poulin, and Barbara Terhal.
\newblock Tradeoffs for reliable quantum information storage in 2d systems.
\newblock {\em Physical review letters}, 104(5):050503, 2010.

\bibitem{baspin2024improved}
Nou{\'e}dyn Baspin, Venkatesan Guruswami, Anirudh Krishna, and Ray Li.
\newblock Improved rate-distance trade-offs for quantum codes with restricted connectivity.
\newblock {\em Quantum Science and Technology}, 10(1):015021, 2024.

\bibitem{dai2024locality}
Samuel Dai and Ray Li.
\newblock Locality vs quantum codes.
\newblock {\em arXiv preprint arXiv:2409.15203}, 2024.

\bibitem{aliferis2007subsystem}
Panos Aliferis and Andrew~W Cross.
\newblock Subsystem fault tolerance with the bacon-shor code.
\newblock {\em Physical review letters}, 98(22):220502, 2007.

\bibitem{bravyi2012subsystem}
Sergey Bravyi, Guillaume Duclos-Cianci, David Poulin, and Martin Suchara.
\newblock Subsystem surface codes with three-qubit check operators.
\newblock {\em arXiv preprint arXiv:1207.1443}, 2012.

\bibitem{higgott2021subsystem}
Oscar Higgott and Nikolas~P Breuckmann.
\newblock Subsystem codes with high thresholds by gauge fixing and reduced qubit overhead.
\newblock {\em Physical Review X}, 11(3):031039, 2021.

\bibitem{hastings2021dynamically}
Matthew~B Hastings and Jeongwan Haah.
\newblock Dynamically generated logical qubits.
\newblock {\em Quantum}, 5:564, 2021.

\bibitem{higgott2024constructions}
Oscar Higgott and Nikolas~P Breuckmann.
\newblock Constructions and performance of hyperbolic and semi-hyperbolic floquet codes.
\newblock {\em PRX Quantum}, 5(4):040327, 2024.

\bibitem{hastings2021quantum}
Matthew~B Hastings.
\newblock On quantum weight reduction.
\newblock {\em arXiv preprint arXiv:2102.10030}, 2021.

\bibitem{sabo2024weight}
Eric Sabo, Lane~G Gunderman, Benjamin Ide, Michael Vasmer, and Guillaume Dauphinais.
\newblock Weight-reduced stabilizer codes with lower overhead.
\newblock {\em PRX Quantum}, 5(4):040302, 2024.

\bibitem{he2025discovering}
Austin~Yubo He and Zi-Wen Liu.
\newblock Discovering highly efficient low-weight quantum error-correcting codes with reinforcement learning.
\newblock {\em arXiv preprint arXiv:2502.14372}, 2025.

\bibitem{ruiz2025ldpc}
Diego Ruiz, J{\'e}r{\'e}mie Guillaud, Anthony Leverrier, Mazyar Mirrahimi, and Christophe Vuillot.
\newblock Ldpc-cat codes for low-overhead quantum computing in 2d.
\newblock {\em Nature Communications}, 16(1):1040, 2025.

\bibitem{lidar2013quantum}
Daniel~A Lidar and Todd~A Brun.
\newblock {\em Quantum error correction}.
\newblock Cambridge university press, 2013.

\bibitem{Note2}
For the qudit version, or generally local-dimension-invariant form, we simply put a minus sign in front of one term so that the symplectic product is perfectly zero.

\bibitem{tillich2013quantum}
Jean-Pierre Tillich and Gilles Z{\'e}mor.
\newblock Quantum ldpc codes with positive rate and minimum distance proportional to the square root of the blocklength.
\newblock {\em IEEE Transactions on Information Theory}, 60(2):1193--1202, 2013.

\bibitem{manes2025distance}
Argyris~Giannisis Manes and Jahan Claes.
\newblock Distance-preserving stabilizer measurements in hypergraph product codes.
\newblock {\em Quantum}, 9:1618, 2025.

\bibitem{golowich2024decoding}
Louis Golowich and Venkatesan Guruswami.
\newblock Decoding quasi-cyclic quantum ldpc codes.
\newblock In {\em 2024 IEEE 65th Annual Symposium on Foundations of Computer Science (FOCS)}, pages 344--368. IEEE, 2024.

\bibitem{yoshida2013exotic}
Beni Yoshida.
\newblock Exotic topological order in fractal spin liquids.
\newblock {\em Physical Review B—Condensed Matter and Materials Physics}, 88(12):125122, 2013.

\bibitem{nixon2021correcting}
Georgia~M Nixon and Benjamin~J Brown.
\newblock Correcting spanning errors with a fractal code.
\newblock {\em IEEE Transactions on Information Theory}, 67(7):4504--4516, 2021.

\bibitem{williamson2024layer}
Dominic~J Williamson and Nou{\'e}dyn Baspin.
\newblock Layer codes.
\newblock {\em Nature Communications}, 15(1):9528, 2024.

\end{thebibliography}

\end{document}